\documentclass[a4paper,onecolumn,superscriptaddress,11pt,accepted=2018-07-16]{quantumarticle}

\usepackage[normalem]{ulem} 

\usepackage{ifthen}

\usepackage{fullpage} 
\usepackage{amsmath} 
\usepackage{amssymb} 
\usepackage{amsthm} 
\usepackage{bbm} 
\usepackage{tabularx} 
\usepackage{graphicx} 
\usepackage{ifpdf}
\ifpdf
\usepackage[ 
  pdftex,
  bookmarks,
  pagebackref,
  plainpages=false, 
  pdfpagelabels=true 
]{hyperref}
\else \fi

\newcommand{\inner}[2]{\langle #1 , #2\rangle}
\newcommand{\Inner}[2]{\left\langle #1 , #2\right\rangle}

\DeclareMathOperator{\fid}{\mathrm{F}}
\DeclareMathOperator{\rfid}{\mathrm{F}_{\mathrm{r}}}

\newcommand{\kprod}[3]{#1_{#2\dots#3}}

\newtheorem{theorem}{Theorem}
\newtheorem{lemma}[theorem]{Lemma}

\newtheorem{proposition}[theorem]{Proposition}

\newtheorem{remark}[theorem]{Remark}
\newtheorem{definition}[theorem]{Definition}

\newcommand{\pa}[1]{(#1)}
\newcommand{\Pa}[1]{\left(#1\right)}

\newcommand{\Br}[1]{\left[#1\right]}

\newcommand{\Set}[1]{\left\{#1\right\}}

\newcommand{\bra}[1]{\langle#1|}

\newcommand{\ket}[1]{|#1\rangle}

\newcommand{\kb}[1]{\ket{#1} \bra{#1}}
\newcommand{\braket}[2]{  \langle #1 \vert #2 \rangle  }

\newcommand{\vecbra}[1]{\langle\hspace{-2.5pt}\langle#1|}
\newcommand{\vecket}[1]{|#1\rangle\hspace{-2.5pt}\rangle}
\newcommand{\veckb}[1]{\vecket{#1} \vecbra{#1}}
\newcommand{\vecbraket}[2]{  \langle\hspace{-2.5pt}\langle #1 \vert #2 \rangle\hspace{-2.5pt}\rangle  }

\def\Jamiolkowski{J}
\newcommand{\jam}[1]{\Jamiolkowski\pa{#1}}

\DeclareMathOperator{\trace}{Tr}
\newcommand{\ptr}[2]{\trace_{#1}\pa{#2}}
\newcommand{\Ptr}[2]{\trace_{#1}\Pa{#2}}
\newcommand{\tr}[1]{\ptr{}{#1}}

\newcommand{\tinyspace}{\mspace{1mu}}
\newcommand{\abs}[1]{|\tinyspace#1\tinyspace|}
\newcommand{\Abs}[1]{\left|\tinyspace#1\tinyspace\right|}
\newcommand{\norm}[1]{\lVert\tinyspace#1\tinyspace\rVert}
\newcommand{\Norm}[1]{\left\lVert\tinyspace#1\tinyspace\right\rVert}

\newcommand{\tnorm}[1]{\norm{#1}_{\trace}}
\newcommand{\Tnorm}[1]{\Norm{#1}_{\trace}}

\newcommand{\dnorm}[1]{\norm{#1}_{\diamond}}

\newcommand{\snorm}[2]{\norm{#1}_{\diamond{#2}}}
\newcommand{\Snorm}[2]{\Norm{#1}_{\diamond{#2}}}

\newcommand{\fontmapset}{\mathbf}

\newcommand{\mset}[2]{\fontmapset{#1}\pa{#2}}

\newcommand{\lin}[1]{\mset{L}{#1}}

\newcommand{\uni}[1]{\mset{U}{#1}}

\newcommand{\her}[1]{\mset{Her}{#1}}

\newcommand{\pos}[1]{\mset{Pos}{#1}}

\newcommand{\den}[1]{\mset{Dens}{#1}}

\newcommand{\uball}[1]{\mset{K}{#1}}

\newcommand{\PBBC}{\mathrm{B_{BC}}}
\newcommand{\PABC}{\mathrm{A_{BC}}}
\newcommand{\PBOT}{\mathrm{B_{OT}}}
\newcommand{\PAOT}{\mathrm{A_{OT}}}

\def\ot{\otimes}

\def\cM{\mathcal{M}}

\def\cS{\mathcal{S}}

\def\cW{\mathcal{W}}
\def\cX{\mathcal{X}}
\def\cY{\mathcal{Y}}
\def\cZ{\mathcal{Z}}

\newcommand{\half}{\frac{1}{2}}

\usepackage{color,graphicx}

\newcommand{\finalMix}[2]{\rho_{#1}({#2})}
\newcommand{\finalMixS}{\finalMix{S}{\tilde B}}
\newcommand{\finalMixT}{\finalMix{T}{\tilde B}} 
\newcommand{\comp}{\circ}

\usepackage{tikz} 
\usetikzlibrary{shapes.geometric}
\usetikzlibrary{shapes.misc}
\usetikzlibrary{decorations.pathmorphing}
\usetikzlibrary{decorations.pathreplacing}
\usetikzlibrary{positioning,backgrounds,shadows}
\usetikzlibrary{fadings}

\usetikzlibrary{calc, trees, positioning, arrows, shapes, shapes.multipart, shadows, matrix, decorations.pathreplacing, decorations.pathmorphing}

\usetikzlibrary{shapes.geometric}
\usetikzlibrary{shapes.misc}
\usetikzlibrary{decorations.pathmorphing}
\usetikzlibrary{decorations.pathreplacing}
\usetikzlibrary{positioning,backgrounds,shadows}
\usetikzlibrary{decorations.markings} 

\begin{document}

\title{
Fidelity of quantum strategies with applications to cryptography
}

\author{Gus Gutoski}
\affiliation{Perimeter Institute for Theoretical Physics, ON, Canada}
\author{Ansis Rosmanis}
\affiliation{Centre for Quantum Technologies, National University of Singapore, Singapore}
\affiliation{School of Physical and Mathematical Sciences, Nanyang Technological University, Singapore}
\author{Jamie Sikora}
\affiliation{Centre for Quantum Technologies, National University of Singapore, Singapore}
\affiliation{MajuLab, \textup{CNRS-UNS-NUS-NTU} International Joint Research Unit, UMI 3654, Singapore}
 
\date{August 28, 2018}
 
\maketitle

\begin{abstract}
We introduce a definition of the fidelity function for multi-round quantum strategies, which we call the \emph{strategy fidelity}, that is a generalization of the fidelity function for quantum states. 
We provide many properties of the strategy fidelity including a Fuchs-van de Graaf relationship with the strategy norm.  
We also provide a general monotonicity result for both the strategy fidelity and strategy norm under the actions of strategy-to-strategy linear maps.  
We illustrate an operational interpretation of the strategy fidelity in the spirit of Uhlmann's Theorem and discuss 
its application to the security analysis of quantum protocols for interactive cryptographic tasks 
such as bit-commitment and oblivious string transfer. 
Our analysis is general in the sense that the actions of the protocol need not be fully specified, which is in stark contrast to most other security proofs. 
Lastly, we provide a semidefinite programming formulation of the strategy fidelity. 
\end{abstract} 

\section{Introduction}
\label{sec:review} 

\subsection{Review of quantum strategies}
       
In this paper we consider multiple-round interactions between two parties involving the exchange of quantum information.
There is a natural asymmetry between the parties as only one of the parties can send the first message or receive the final message.
Since we are not concerned about optimizing the number of messages exchanged, without loss of generality both of these tasks are done by the same party, which, for convenience, we call \emph{Bob}.
Let us call the other party \emph{Alice}.
The interaction between Alice and Bob decomposes naturally into a finite number $r$ of \emph{rounds} (see Figure\ \ref{fig:interaction}). 
 
Such interactions are conveniently described by the formalism of quantum strategies
introduced in Ref.~\cite{GutoskiW07}.
We closely follow that formalism here with the exception that we consider two mathematically different objects: \emph{strategies} and \emph{pure strategies}.
Pure strategies are implemented using linear isometries and preserve their final memory space,
while strategies trace out the final memory space. 
The object we call a strategy is called a \emph{non-measuring strategy} in Ref.~\cite{GutoskiW07}. 
 For additional details on quantum strategies,
one may refer to~\cite{GutoskiW07,ChiribellaD+09a,Gutoski-Phd}.

\begin{definition}[Pure strategy and pure co-strategy] \label{def:purestrat}
Let $r\geq 1$ and let $\cX_1,\ldots,\cX_r,\cY_1,\ldots,\cY_r,\cZ_r$, $\cW_r$ be complex Euclidean spaces
and, for notational convenience, let $\cX_{r+1} := \mathbb{C}$ and $\cZ_0 := \mathbb{C}$.
An \emph{$r$-round pure strategy} $\tilde A$ having \emph{input spaces} $\cX_1,\ldots,\cX_r$,  \emph{output spaces} $\cY_1,\ldots,\cY_r$, and \emph{final memory space} $\cZ_r$, consists of:
\begin{enumerate}
\item[1.]
  complex Euclidean spaces $\cZ_1,\ldots,\cZ_{r-1}$, called \emph{intermediate memory spaces}, and
\item[2.]
  an $r$-tuple of linear isometries $(A_1,\ldots,A_r)$ of the form
  $A_i : \cX_i\ot\cZ_{i-1}\to \cY_i\ot\cZ_i$.
\end{enumerate} 

An \emph{$r$-round pure co-strategy} having input spaces $\cY_1,\ldots,\cY_r$, output spaces $\cX_1,\ldots,\cX_r$, and final memory space $\cW_r$,  consists of:
\begin{enumerate}
\item[1.]
  complex Euclidean intermediate memory spaces $\cW_0,\ldots,\cW_{r-1}$,
\item[2.]
  a pure quantum state $\ket{\beta} \in \cX_1\otimes\cW_0$, called the \emph{initial state}, and
\item[3.]
  an $r$-tuple of linear isometries $(B_1,\ldots,B_r)$ of the form
  $B_i : \cY_i\ot\cW_{i-1}\to \cX_{i+1}\ot\cW_i$. 
\end{enumerate}
A pure strategy and a pure co-strategy are said to be \emph{compatible} when
the input spaces of one are the output spaces of the other, and vice versa.
The \emph{final state} after the interaction between $\tilde A$ and $\tilde B$ is denoted by 
\[
\ket{\psi(\tilde A,\tilde B)} := (I_{\cZ_r}\ot B_r)(A_r\ot I_{\cW_{r-1}})\cdots (I_{\cZ_{1}}\ot B_{1})(A_{1}\ot I_{\cW_{0}}) \ket{\beta} \in\cZ_r\ot\cW_r. 
\] 
\end{definition}  
In order to extract classical information from the interaction it suffices to permit Alice and Bob to measure their respective parts of the final state $\ket{\psi(\tilde A,\tilde B)}$.

\newcommand{\xSep}{1.25cm} 
\newcommand{\ySep}{1.3cm} 
\newcommand{\xFirst}{2.8cm} 
\newcommand{\xLast}{13.6cm} 
\newcommand{\boxSize}{1.0cm} 
\newcommand{\arrowGap}{0.1cm} 
\newcommand{\arrowLabelYGap}{0.22cm} 
\newcommand{\diagLabelXGap}{0.13cm} 
\newcommand{\diagLabelYGap}{0.35cm} 

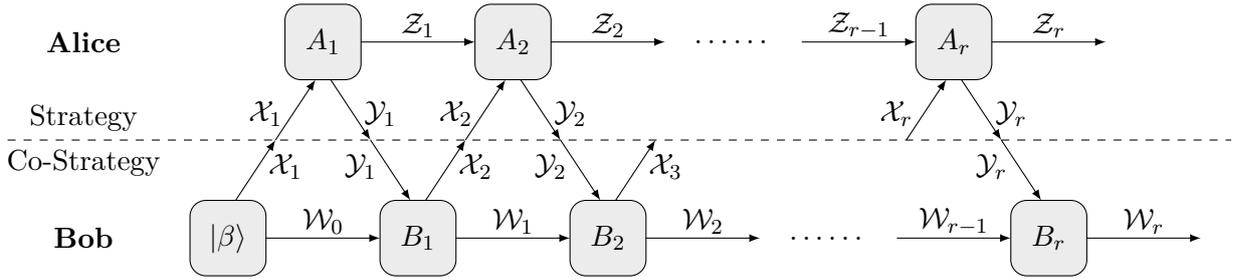
\begin{figure}[tbh]
\centering
\begin{tikzpicture}
\draw [dashed] (-0.1,0) -- + (16.1,0);
\node at (0.9, 0.3) {Strategy};
\node at (0.9,-0.3) {Co-Strategy};
\node at (0.9, \ySep) {\bf{Alice}};
\node at (0.9,-\ySep) {\bf{Bob}};
\tikzstyle{isom} = [	draw,
						fill=lightgray!30!white,
						minimum height=\boxSize,
						minimum width=\boxSize,
						align=center,
						rounded corners = 6pt]
\node [isom] at (\xFirst+1*\xSep, \ySep) {$A_1$};
\node [isom] at (\xFirst+3*\xSep, \ySep) {$A_2$};
\node [isom] at (\xLast-1*\xSep, \ySep) {$A_r$};

\node [isom] at (\xFirst,             -\ySep) {$\ket{\beta}$};
\node [isom] at (\xFirst+2*\xSep,-\ySep) {$B_1$};
\node [isom] at (\xFirst+4*\xSep,-\ySep) {$B_2$};
\node [isom] at (\xLast,             -\ySep) {$B_r$};

\tikzset{>=latex} 
\draw [->] (\xFirst+0.5*\xSep,0) -- +( 0.5*\xSep-\arrowGap, \ySep-\boxSize/2);
\draw [<-] (\xFirst+0.5*\xSep,0) -- +(-0.5*\xSep+\arrowGap,-\ySep+\boxSize/2);
\draw [->] (\xFirst+2.5*\xSep,0) -- +( 0.5*\xSep-\arrowGap, \ySep-\boxSize/2);
\draw [<-] (\xFirst+2.5*\xSep,0) -- +(-0.5*\xSep+\arrowGap,-\ySep+\boxSize/2);
\draw [<-] (\xFirst+4.5*\xSep,0) -- +(-0.5*\xSep+\arrowGap,-\ySep+\boxSize/2);
\draw [->] (\xLast-1.5*\xSep,0) -- +( 0.5*\xSep-\arrowGap, \ySep-\boxSize/2);
\draw [->] (\xFirst+1.5*\xSep,0) -- +( 0.5*\xSep-\arrowGap,-\ySep+\boxSize/2);
\draw [<-] (\xFirst+1.5*\xSep,0) -- +(-0.5*\xSep+\arrowGap, \ySep-\boxSize/2);
\draw [->] (\xFirst+3.5*\xSep,0) -- +( 0.5*\xSep-\arrowGap,-\ySep+\boxSize/2);
\draw [<-] (\xFirst+3.5*\xSep,0) -- +(-0.5*\xSep+\arrowGap, \ySep-\boxSize/2);
\draw [->] (\xLast-0.5*\xSep,0) -- +( 0.5*\xSep-\arrowGap,-\ySep+\boxSize/2);
\draw [<-] (\xLast-0.5*\xSep,0) -- +(-0.5*\xSep+\arrowGap, \ySep-\boxSize/2);
\draw [->] (\xFirst             +\boxSize/2,-\ySep) -- +(2*\xSep-\boxSize,0);
\draw [->] (\xFirst+2*\xSep+\boxSize/2,-\ySep) -- +(2*\xSep-\boxSize,0);
\draw [->] (\xFirst+4*\xSep+\boxSize/2,-\ySep) -- +(2*\xSep-\boxSize,0);
\draw [->] (\xLast-2*\xSep+\boxSize/2,-\ySep) -- +(2*\xSep-\boxSize,0);
\draw [->] (\xLast            +\boxSize/2,-\ySep) -- +(2*\xSep-\boxSize,0);
\draw [->] (\xFirst+1*\xSep+\boxSize/2,\ySep) -- +(2*\xSep-\boxSize,0);
\draw [->] (\xFirst+3*\xSep+\boxSize/2,\ySep) -- +(2*\xSep-\boxSize,0);
\draw [->] (\xLast-3*\xSep+\boxSize/2,\ySep) -- +(2*\xSep-\boxSize,0);
\draw [->] (\xLast-1*\xSep+\boxSize/2,\ySep) -- +(2*\xSep-\boxSize,0);

\node at (\xFirst+1*\xSep,-\ySep+\arrowLabelYGap) {$\cW_0$};
\node at (\xFirst+3*\xSep,-\ySep+\arrowLabelYGap) {$\cW_1$};
\node at (\xFirst+5*\xSep,-\ySep+\arrowLabelYGap) {$\cW_2$};
\node at (\xLast-1*\xSep,-\ySep+\arrowLabelYGap) {$\cW_{r-1}$};
\node at (\xLast+1*\xSep,-\ySep+\arrowLabelYGap) {$\cW_r$};

\node at (\xFirst+2*\xSep,\ySep+\arrowLabelYGap) {$\cZ_1$};
\node at (\xFirst+4*\xSep,\ySep+\arrowLabelYGap) {$\cZ_2$};
\node at (\xLast-2*\xSep,\ySep+\arrowLabelYGap) {$\cZ_{r-1}$};
\node at (\xLast            ,\ySep+\arrowLabelYGap) {$\cZ_r$};

\node at (\xFirst+0.5*\xSep-\diagLabelXGap, \diagLabelYGap) {$\cX_1$};
\node at (\xFirst+0.5*\xSep+\diagLabelXGap,-\diagLabelYGap) {$\cX_1$};
\node at (\xFirst+2.5*\xSep-\diagLabelXGap, \diagLabelYGap) {$\cX_2$};
\node at (\xFirst+2.5*\xSep+\diagLabelXGap,-\diagLabelYGap) {$\cX_2$};
\node at (\xFirst+4.5*\xSep+\diagLabelXGap,-\diagLabelYGap) {$\cX_3$};
\node at (\xLast-1.5*\xSep-\diagLabelXGap,\diagLabelYGap) {$\cX_r$};

\node at (\xFirst+1.5*\xSep-\diagLabelXGap,-\diagLabelYGap) {$\cY_1$};
\node at (\xFirst+1.5*\xSep+\diagLabelXGap, \diagLabelYGap) {$\cY_1$};
\node at (\xFirst+3.5*\xSep-\diagLabelXGap,-\diagLabelYGap) {$\cY_2$};
\node at (\xFirst+3.5*\xSep+\diagLabelXGap, \diagLabelYGap) {$\cY_2$};
\node at (\xLast-0.5*\xSep-\diagLabelXGap,-\diagLabelYGap) {$\cY_r$};
\node at (\xLast-0.5*\xSep+\diagLabelXGap, \diagLabelYGap) {$\cY_r$};

\node at (\xFirst/2+2*\xSep+\xLast/2,-\ySep) {$\dots\dots$};
\node at (\xFirst/2+1*\xSep+\xLast/2,\ySep) {$\dots\dots$};
\end{tikzpicture}
\caption{An $r$-round interaction between a pure strategy of Alice (the linear isometries above the dashed line) and a pure co-strategy of Bob (the linear isometries below the dashed line).
  Arrows crossing the dashed line represent messages exchanged between the parties,
  while horizontal arrows represent private memory. 
}
\label{fig:interaction}
\end{figure}
  
A pure strategy $\tilde A$ specified by linear isometries $(A_1,\ldots,A_r)$ can be represented by a single isometry
\begin{equation}
\tilde A :=
(A_r\ot I_{\kprod{\cY}{1}{r-1}})
\dots
(I_{\kprod{\cX}{3}{r}}\ot A_2\ot I_{\cY_1})
(I_{\kprod{\cX}{2}{r}}\ot A_1)
: \kprod{\cX}{1}{r} \to \kprod{\cY}{1}{r}\ot\cZ_r,
\end{equation}
{where $\kprod{\cX}{i}{j}$ is short for $\cX_i\ot\cdots\ot\cX_j$ and $\kprod{\cY}{i}{j}$ is short for $\cY_i\ot\cdots\ot\cY_j$.}
We abuse the notation\footnote{It will be clear from context to which we are referring.} $\tilde A$ here and elsewhere in the paper by using it to denote both
a pure strategy and the linear isometry representing it, and we do the same for pure co-strategies $\tilde B$, discussed next. 
A pure co-strategy $\tilde B$ specified by the initial state $\ket{\beta}$ and linear isometries $(B_1,\ldots,B_r)$  can be represented by a single isometry
\begin{equation}
\tilde B :=
(B_r\ot I_{\kprod{\cX}{1}{r}})
\cdots
(I_{\kprod{\cY}{2}{r}}\ot B_1\ot I_{\cX_1})
(I_{\kprod{\cY}{1}{r}}\ot \ket{\beta})
: \kprod{\cY}{1}{r}\to\kprod{\cX}{1}{r}\ot\cW_r.
\end{equation}
Note that two pure strategies that are represented by the same linear isometry are effectively
indistinguishable, and the same holds true for pure co-strategies. 
       
After the interaction, Alice's actions do not affect Bob's
reduced state (and vice versa).
Hence, from Bob's point of view, Alice can trace out her final memory space. 
In view of this, a \emph{strategy} $A$ is obtained from a pure
strategy $\tilde A$ by tracing out the final memory space $\cZ_r$ and
a \emph{co-strategy} $B$ is obtained from a pure co-strategy $\tilde
B$ by tracing out the final memory space $\cW_r$.%
\footnote{We note that, technically, strategies and pure
strategies are (slightly) different mathematical objects.}
Multiple pure strategies (pure co-strategies) can yield the same strategy
(co-strategy), and we call any such pure strategy (pure co-strategy) a
\emph{purification}.
We will use tildes to indicate purifications.
 
Just as a pure strategy and a pure co-strategy can be specified by linear isometries $\tilde A$ and $\tilde B$, respectively, their corresponding strategy $A$ and co-strategy $B$ can be specified by quantum channels
\begin{alignat}{3}
 &\Phi_A \,&&: \lin{\kprod{\cX}{1}{r}}\to\lin{\kprod{\cY}{1}{r}} \,&&: X\mapsto \ptr{\cZ_r}{\tilde AX\tilde A^*}, \\
 &\Psi_B \,&&: \lin{\kprod{\cY}{1}{r}}\to\lin{\kprod{\cX}{1}{r}} \,&&: Y\mapsto \ptr{\cW_r}{\tilde BY\tilde B^*},
\end{alignat}
where $\lin{\cX}$ is the set of all linear operators acting on a space $\cX$.
In turn, both of these channels can be specified using their Choi-Jamio{\l}kowski representations, but, due to the asymmetry between strategies and co-strategies, it is convenient to specify the latter one using the Choi-Jamio{\l}kowski representation of its adjoint map.
Thus, we can represent a strategy $A$ by $\jam{\Phi_A}$ and a co-strategy $B$ by $\jam{\Psi_B^*}$, both of which are positive semidefinite operators acting on $\kprod{\cY}{1}{r}\ot\kprod{\cX}{1}{r}$.
In a similar abuse of notation as mentioned before, we refer to $\jam{\Phi_A}$ as the strategy $A$ and to $\jam{\Psi_B^*}$ as the co-strategy $B$. 
 
For compatible pure strategy $\tilde A$ and pure co-strategy $\tilde B$, let
\begin{equation} \label{eq:finalMixed}
\finalMix{A}{\tilde B} := \Ptr{\cZ_r}{\kb{\psi(\tilde A,\tilde B)}} 
\end{equation}
denote the reduced state of the final memory space $\cW_r$ of $\tilde B$ after
the interaction between $\tilde A$ and $\tilde B$. 
Since this state is the same for all purifications of $A$, we omit the tilde above $A$ in this notation.
 
\subsection{The definition of strategy fidelity}

Recall that the fidelity $\fid(P,Q)$ between two positive semidefinite operators $P$ and $Q$ is defined as 
\begin{equation}
  \fid(P,Q) := \Tnorm{\sqrt{P}\sqrt{Q}}.
\end{equation} 
When applied to density operators $\rho,\xi$, the fidelity function $\fid(\rho,\xi)$ is a useful distance measure for quantum states.
We would like to construct a generalization of the fidelity function that can serve as a useful distance measure for quantum strategies. 
  
Just as the trace norm $\tnorm{\rho-\xi}$ quantifies the distinguishability of quantum states, the strategy norm\footnote{What we refer to as the strategy norm was introduced in \cite{ChiribellaD+08b} where it is called the \emph{operational norm}. We use the term \emph{strategy norm} to make the connections to strategy fidelity more apparent.} $\snorm{S-T}{\mathrm{r}}$ studied in~\cite{ChiribellaD+08b}~and~\cite{Gutoski12}, quantifies the distinguishability of quantum strategies $S$ and 
$T$ having the same input and output spaces. 
In other words, ${\snorm{S-T}{\mathrm{r}}}$ is proportional to the maximum bias with which an interacting pure co-strategy $\tilde B$ can distinguish $S$ from $T$. Another expression for this maximum bias can be derived as follows.
Let $\cW_r$ be the final memory space of $\tilde B$ and let $\finalMixS,\finalMixT$ be
the reduced states of this final memory space after an interaction between $\tilde B$ and $S,T$, respectively, as defined in~\eqref{eq:finalMixed}. 
It is clear that the maximum bias with which $S$ can be distinguished from $T$ is proportional to the maximum over all such $\tilde B$ with which the final state $\finalMixS$ can be distinguished from $\finalMixT$, which is precisely $\tnorm{\finalMixS-\finalMixT}$.

\begin{remark}
\label{rem:UniInvariance} 
All purifications $\tilde B$ of $B$ are equivalent up to a unitary acting on $\cW_r$.
Thus, unitarily invariant distance measures between $\finalMixS$ and $\finalMixT$ (including the trace distance and the fidelity) depend only upon $B$ and not upon the specific purification $\tilde B$. 
\end{remark}
The strategy norm is defined (see Definition~\ref{def:str_norm}) so that
\begin{equation} \label{snorm}
  \snorm{S-T}{\mathrm{r}} = {\max_B} \, \tnorm{\finalMixS-\finalMixT}.
\end{equation}
In light of this observation, we define the strategy fidelity by replacing the maximization of the trace distance  between $\finalMixS$ and $\finalMixT$ with the minimization of the fidelity between $\finalMixS$ and $\finalMixT$. 

\begin{definition}[Strategy fidelity] \label{def:strat-fid}
  For any $r$-round strategies $S$ and $T$ having the same input and output spaces, the \emph{strategy fidelity} is defined as  
  \begin{equation} \label{eq:naive-def}
    \rfid(S,T) := \min_B \, \fid(\finalMixS,\finalMixT) 
  \end{equation}
where the minimization is over all compatible co-strategies $B$ and the states $\finalMixS,\finalMixT$ are as defined in~\eqref{eq:finalMixed}.
\end{definition} 

In the following discussion, we argue that this definition is a meaningful one by proving analogues of the Fuchs-van de Graaf inequalities and Uhlmann's Theorem for the strategy fidelity, among many other properties. 
 
\begin{remark} 
The same definition of fidelity has been considered for the case of channels~\textup{\cite{10.1063/1.1904510}}. In that setting, they establish several properties which we generalize to the strategy setting.  
\end{remark} 

First, let us observe that the fidelity for quantum states is recovered as a special case of the strategy fidelity when $S,T$ are one-round strategies with no input (that is, $\cX_1=\mathbb{C}$) and only one output message.
To see this, observe that one-round strategies such as $S,T$ are simply states $\rho,\xi$ acting on $\cY_1$.
Bob's most general pure co-strategy is an isometry $\tilde B:\cY_1\to\cW_1$.  
In this case the effect of Bob's purified strategy $\tilde B$ is cancelled in the computation of $\rfid(S,T)$ so that
\begin{equation}
  \fid_1(S,T) = \min_B \, \fid(\finalMixS,\finalMixT) = \fid(\tilde B\rho \tilde B^*, \tilde B\xi \tilde B^*) = \fid(\rho,\xi)
\end{equation} 
as claimed.
    
\paragraph{Basic properties of the strategy fidelity} \quad  \\ 
  
We now list several other properties of the strategy fidelity, all of which immediately hold using the corresponding properties of the fidelity of quantum states (see references~\cite{FvdG99, NC00}).  

\begin{proposition}[Basic properties] \quad 
\label{prop:basicProp}
\begin{itemize} 
\item (Fuchs-van de Graaf inequalities for strategies) 
For any $r$-round strategies $S$ and $T$, it holds that
  \begin{equation} \label{FvdG} 
    1 - \frac{1}{2}\snorm{S-T}{\mathrm{r}} \leq \rfid(S,T) \leq \sqrt{1 - \frac{1}{4}\snorm{S-T}{\mathrm{r}}^2}.
  \end{equation}
\item (Symmetry) 
For any $r$-round strategies $S$ and $T$, it holds that 
$\rfid(S,T) = \rfid(T,S)$.
\item (Joint concavity) 
For any $r$-round strategies $S^1, \ldots, S^n$ and $T^1, \ldots, T^n$, and nonnegative scalars $\lambda_1, \ldots, \lambda_n$ satisfying $\sum_{i=1}^n \lambda_i = 1$, we have
\begin{equation}
 \rfid \left( \sum_{i=1}^n \lambda_i S^i, \sum_{i=1}^n \lambda_i T^i \right) \geq \sum_{i=1}^n \lambda_i \rfid \left( S^i, T^i \right).
\end{equation}
\item (Bounds on the strategy fidelity) 
For any $r$-round strategies $S$ and $T$, we have
${0 \leq \rfid(S,T) \leq 1}$. 
Moreover, 
$\rfid(S,T) = 1$ if and only if $S = T$
and
$\rfid(S,T) = 0$ if and only if $S$ and $T$ are perfectly distinguishable.
\end{itemize}
\end{proposition}  

We later discuss that the strategy version of the Fuchs-van de Graaf inequalities is crucial to our cryptographic applications. This was also used implicitly in~\cite{ChiribellaD+09b}. 

\paragraph{Monotonicity of the strategy fidelity and the strategy norm} \quad \\ 

The fidelity for quantum states is known to be monotonic under channels, meaning that
\begin{equation}
  \fid(\Phi(\rho),\Phi(\xi)) \geq \fid(\rho,\xi)
\end{equation}
for any choice of states $\rho,\xi$ and channel $\Phi$ \cite{BarnumC+96}. 
It was observed in Ref.~\cite{10.1063/1.1904510} 
that the fidelity function of \emph{quantum channels} (that aligns with our definition of strategy fidelity for a $1$-round interaction) is also monotonic under composition (both left and right) with another channel. That is, 
\begin{equation}
  \fid_1(\Phi \comp\Delta,\Psi \comp\Delta) \geq \fid_1(\Phi,\Psi)
  \qquad \textnormal{and} \qquad
  \fid_1(\Delta' \comp\Phi,\Delta' \comp\Psi) \geq \fid_1(\Phi,\Psi)
\end{equation}
for all channels $\Phi,\Psi:\lin{\cX}\to\lin{\cY}$ and $\Delta$ into $\lin{\cX}$ and $\Delta'$ on $\lin{\cY}$.
However, there are other physical maps on channels that cannot in general be written as a composition with another channel. 
Chiribella, D'Ariano, and Perinotti call such mappings \emph{supermaps} and characterize them in Ref.~\cite{ChiribellaD+08f}.  
Thus, the natural generalization of monotonicity of the kind described above 
would be the analogous statement involving supermaps.
We provide an even stronger result concerning monotonicity of the strategy fidelity using the following definition.
 
\begin{definition}
A \emph{strategy supermap} is a completely positive linear map (with respect to Choi-Jamio{\l}kowski representations) that maps $r$-round strategies
to $r'$-round strategies. It is understood that $r$-round strategies are for some choice of input spaces $\cX_1,\dots,\cX_r$ and output spaces $\cY_1,\dots,\cY_r$ and $r'$-round strategies are for some choice of input spaces $\cX'_1,\dots,\cX'_{r'}$ and output spaces $\cY'_1,\dots,\cY'_{r'}$. 
\end{definition}

The definition of strategy supermaps are inspired by 
physically realizable maps from $r$-round strategies to $r'$-round 
strategies studied by Chiribella, D'Ariano, and Perinotti~\cite{ChiribellaD+09a}.
Our result, however, is purely mathematical and does not require strategy supermaps to be physically realizable.

\begin{theorem}[Monotonicity of the strategy fidelity]
\label{mon-fid}
  For all natural numbers $r,r'$, all $r$-round strategies $S,T$, and all strategy supermaps $\Upsilon$ from $r$-round strategies to $r'$-round strategies, it holds that
  \begin{equation}
    \fid_{\mathrm{r'}}(\Upsilon(S),\Upsilon(T)) \geq \rfid(S,T).
  \end{equation} 
\end{theorem} 

We can also prove a similar monotonicity result for the strategy norm.
By analogy with the fidelity, the trace norm is known to be monotonic under channels, meaning that
\begin{equation}
  \tnorm{\Phi(X)}\leq\tnorm{X}
\end{equation}
for all operators $X$ and all channels $\Phi$ \cite{Ruskai94}.
Similarly, the diamond norm can be shown to be monotonic under composition (both left and right) with channels, meaning that
\begin{equation}
  \dnorm{\Phi\comp\Delta}\leq\dnorm{\Phi}
  \qquad\textnormal{and}\qquad
  \dnorm{\Delta'\comp\Phi}\leq\dnorm{\Phi}
\end{equation}
for all linear maps $\Phi:\lin{\cX}\to\lin{\cY}$ and all channels $\Delta$ into $\lin{\cX}$ and $\Delta'$ on $\lin{\cY}$.
As with the fidelity function for quantum channels, defined in Ref.~\cite{10.1063/1.1904510}, monotonicity of the diamond norm under arbitrary supermaps has not yet been observed, nor has monotonicity of the strategy norm under strategy supermaps.

We now establish a monotonicity result for the strategy norm, defined below. 
\begin{definition}[Strategy norm \cite{ChiribellaD+08b}, \cite{Gutoski12}] \label{def:str_norm} 
Consider $\cX_1,\ldots,\cX_r$ and $\cY_1,\ldots,\cY_r$ as input and output spaces of $r$-round strategies, respectively. The strategy norm of a Hermitian operator $H$ acting on $\kprod{\cY}{1}{r}\ot\kprod{\cX}{1}{r}$ is defined as 
\begin{equation}
\Snorm{H}{\mathrm{r}}:=
\max_{B_0,B_1\succeq 0}\Set{\Inner{B_0-B_1}{H} : \textnormal{$B_0 + B_1$ is an $r$-round co-strategy}},
\end{equation}
where the maximization is over all positive semidefinite operators $B_0,B_1$ acting on $\kprod{\cY}{1}{r}\ot\kprod{\cX}{1}{r}$ such that $B_0 + B_1$ is an $r$-round co-strategy having input spaces $\cY_1,\ldots,\cY_r$ and output spaces $\cX_1,\ldots,\cX_r$.
\end{definition} 
Given $H$ as a difference of two $r$-round strategies $S$ and $T$, Definition~\ref{def:str_norm} implies Eqn.~\eqref{snorm}.
 
\begin{theorem}[Monotonicity of the strategy norm] 
\label{mon-norm}
  For all natural numbers $r,r'$, all Hermitian operators $H$ acting on $\kprod{\cY}{1}{r}\ot\kprod{\cX}{1}{r}$ and strategy supermaps $\Upsilon$ from $r$-round strategies to $r'$-round strategies, it holds that
  \begin{equation}
    \Snorm{\Upsilon(H)}{\mathrm{r'}} \leq \Snorm{H}{\mathrm{r}}. 
  \end{equation}
\end{theorem} 

\paragraph{Operational interpretation (min-max properties)}    \quad \\ 

Here we propose an operationally motivated generalization of Uhlmann's Theorem~\cite{uhl76} to the strategy fidelity. In so doing we elucidate the need for a min-max theorem.  
Recall that Uhlmann's Theorem for quantum states asserts that the fidelity $\fid(\rho,\xi)$ between any two quantum states $\rho$ and $\xi$, acting on $\cX$, is given by
\begin{equation} \label{eq:recall-Uhlmann}
  \fid(\rho,\xi) = \max_U \Abs{\bra{\phi}(U\ot I_\cX)\ket{\psi}}
\end{equation}
where $\ket{\phi},\ket{\psi}\in\cX \otimes \cY$ are any purifications of $\rho,\xi$ and the maximization is over all unitaries $U$ acting on $\cY$ alone. 

Intuitively, $\rfid(S,T)$ should quantify the extent to which any purifications $\tilde S,\tilde T$ of two strategies $S,T$ can be made to look the same by acting only on the final memory space $\cZ_r$.  
It follows immediately from the definition of the strategy fidelity and Uhlmann's Theorem that
\begin{equation} \label{eq:naive-min-max}
  \rfid(S,T) = \min_B \fid(\finalMixS,\finalMixT) = \min_B \max_U \Abs{ \bra{\psi(\tilde S,\tilde B)}\Pa{U\ot I_{\cW_r}}\ket{\psi(\tilde T,\tilde B)} }
\end{equation}
where, again, the maximization is over all unitaries $U$ acting on $\cZ_r$ alone.

Notice the order of minimization and maximization in \eqref{eq:naive-min-max}. This could be viewed as a competitive game between Alice (who plays according to $S$ or $T$) and Bob (who plays according to any arbitrary co-strategy $B$) in which Bob is trying to distinguish $S$ from $T$ and Alice is trying to make $S$ and $T$ look the same.
To these ends, Bob chooses his strategy $B$ so as to minimize the overlap $\abs{\braket{\psi(\tilde S,\tilde B)}{\psi(\tilde T,\tilde B)}}$;
given such a choice $B$ for Bob, Alice's responds with a unitary $U$ that maximizes this overlap.

The problem is that Alice's choice of $U$ may depend upon Bob's co-strategy $B$.
The task of distinguishing $S$ from $T$ should depend only upon $S$ and $T$---Alice should not be granted the ability to tweak $S$ or $T$ after she has acquired knowledge of Bob's specific choice of distinguishing co-strategy $B$.
From an operational perspective, it would be much more desirable if the order of minimization and maximization in \eqref{eq:naive-min-max} were reversed. Alice should select her unitary $U$ so as to make $S$ look as much as possible like $T$ \emph{before} Bob selects his distinguishing co-strategy $B$. Thus, we require a type of \emph{min-max theorem}.
 
The set of all co-strategies $B$ for Bob is compact and convex~\cite{GutoskiW07}, but it is not at all clear that the objective function in \eqref{eq:naive-min-max} is convex in $B$; we show later (Lemma~\ref{lm:linear-in-B}) that this is indeed the case. 
However, the set of all unitaries $U$ for Alice is not a convex set.
One might think that we could extend the domain of maximization to the convex hull of the unitaries in the hopes that there is a saddle point $(U,B)$ with $U$ unitary.
Unfortunately, saddle points do not in general occur at extreme points of the domain, so we are not guaranteed that such a unitary saddle point exists.
Thus, a min-max theorem for the strategy fidelity involving unitaries is not so easily forthcoming. 

However, if we allow Alice to apply a general \emph{quantum channel}, we are able to obtain a min-max result, as stated below.  

\begin{theorem}[Strategy generalization of Uhlmann's Theorem]\label{thm:StrUhlmann}
Let $S,T$ be $r$-round strategies and let $\tilde S,\tilde T$ be any purifications of $S,T$. Let
$\ket{\psi(\tilde S,\tilde B)}$, $\ket{\psi(\tilde T,\tilde B)}$ be as defined in Definition~\ref{def:purestrat}.
We have 
\begin{align}
\rfid(S,T)^2 
&= \max_\Xi \, \min_B \, 
\bra{\psi(\tilde S,\tilde B)} \Br{ \Pa{\Xi\otimes I_{\lin{\cW_r}}}\Pa{\kb{\psi(\tilde T,\tilde B)}} } \ket{\psi(\tilde S,\tilde B)} \label{2ndeqn}
\end{align}
where the minimum is over all $r$-round pure co-strategies $\tilde{B}$ and the maximum is over all  quantum channels $\Xi$ acting on $\cZ_r$ alone. 
\end{theorem} 

Note that similar min-max results are derived in~\textup{\cite{10.1063/1.1904510}} and~\textup{\cite{ChiribellaD+09b}}. 
It will be convenient to define the following quantum channel. 

\begin{definition} \label{def:fidchannel} 
A \emph{strategy fidelity-achieving} channel $\Xi$ is a channel which attains the  maximum in \eqref{2ndeqn}, above. 
\end{definition}

\paragraph{Semidefinite programming formulation of strategy fidelity}    \quad \\ 
 
It was shown in~\cite{Gutoski12} that the strategy norm has a semidefinite programming formulation. Also, the fidelity of quantum states has semidefinite programming formulations, see~\cite{Watrous09, Watrous13} for examples.  It is natural to ask whether the strategy fidelity has such a formulation. We answer this question in the affirmative, below. 

\begin{theorem}[Semidefinite programming formulation of strategy fidelity] 
\label{SDP}
Fix any purifications $\tilde S$ and $\tilde T$ of $r$-round strategies $S$ and $T$, respectively. Then $\rfid(S,T)^2$ is equal to the optimal objective function value of the following semidefinite program: 
\begin{equation}
 \begin{array}{rrrllllll}
\rfid(S,T)^2 = & {\max} & t \\ 
& \textup{subject to} & t I_{\cX_1} & \preceq & \ptr{\cY_1}{R_1}  \\ 
& & R_j \otimes I_{\cX_{j+1}} & \preceq & \ptr{\cY_{j+1}}{R_{j+1}}, \, \textup{ for } j \in \{1, \ldots, r-1 \}, \\  
& & R_r & \preceq & \frac{1}{2} \Ptr{\cZ_r}{\Pa{K \otimes I_{\kprod{\cY}{1}{r}\ot\kprod{\cX}{1}{r}}} \vecket{\tilde T}\vecbra{\tilde S}} + h.c. \\ 
& & \left[ \begin{array}{cc} 
I_{\cZ_r} & K \\ 
K^* & I_{\cZ_r} 
\end{array} \right] & \succeq & 0 
\end{array}
 \end{equation} 
where the variables $R_j$ are Hermitian matrices acting on $\kprod{\cY}{1}{j}\ot\kprod{\cX}{1}{j}$ for each $j \in \{ 1, \ldots, r \}$, and $h.c.$ denotes the Hermitian conjugate. 
Note that the optimization is over the Hermitian matrices $R_1, \ldots, R_r$, the scalar $t$, and a (not neccessarily Hermitian) matrix $K$. (The last constraint requires $K$ to be in the convex hull of the set of unitaries acting on space $\cZ_r$.)
\end{theorem} 

There are a few reasons why it is beneficial to have a semidefinite programming formulation of the strategy fidelity (or any other function for that matter). One is that efficient algorithms that approximate semidefinite programs allow for the calculation of numerical values for specific instances (assuming the problem instance is not too large for the computational platform). Another reason is that semidefinite programming has a rich duality theory, which allows one to certify bounds (upper bounds in this case) on the value of the strategy fidelity. Otherwise, such a task would be very hard using the definition alone.
  
\subsection{Applications to two-party quantum cryptography}    
    
Since the seminal work of Wiesner~\cite{Wiesner83} and Bennett and Brassard~\cite{BB84}, there has been much interest in knowing the advantages, and limitations, of quantum protocols for cryptographic tasks. Due to the interactive setting of such protocols, the use of quantum strategy analysis has proven to be useful. In~\cite{GutoskiW07}, it was shown how to rederive Kitaev's lower bound for coin-flipping~\cite{Kitaev02}. In~\cite{ChiribellaD+09b}, it was shown how to find a simple proof of the impossibility of interactive bit-commitment. Here, we find a similar proof of this and extend the argument to oblivious string transfer. 

In this paper, we present our ideas using the machinery we have developed for the strategy fidelity. 
In particular, we show that the strategy version of the Fuchs-van de Graaf inequalities (Eqn.~\eqref{FvdG}) are of central importance in providing security lower bounds. In fact, due to the nature of the strategy norm and strategy fidelity, we are able to bound the security without even specifying the entire protocol! To the best of our knowledge, this very general setting has only been studied in a few security proofs (in particular, of bit-commitment) \cite{10.1063/1.1904510, DArianoKSW07, ChiribellaD+09b} and is in stark contrast to many other security proofs, for example in~\cite{Kitaev02, SR01, Amb01, NayakS03, ABD+04, KerenidisN04, GutoskiW07, ChaillouxK09, ChaillouxK11, ChaillouxKS13-QIC, ChaillouxKS13, NST15, NST16, CGS16, Sikora17} where Alice and Bob's actions are assumed to be fully specified (and known to cheating parties). 

In this paper, we show the impossibility of ideal quantum protocols for interactive \emph{bit-commitment} and \emph{oblivious string transfer}. 

\paragraph{Interactive bit-commitment} \quad \\  

In bit-commitment, we require Alice and Bob to interact over two communication stages:
\begin{itemize}
\item Commit Phase: Alice chooses a uniformly random bit $a$ and interacts with Bob using an $r$-round pure strategy $\tilde A^a$.  
\item Reveal Phase: Alice sends $a$ to Bob and continues her interaction\footnote{Note that the interaction of the Reveal Phase is not part of the strategy $\tilde A^a$. In fact, our results do not depend on the structure of the Reveal Phase, other than revealing $a$.} with him (so that Bob can test if she has cheated). 
\item Cheat Detection: Bob, knowing which pure strategy $\tilde  B$ he has used, measures to check if the final state is consistent with Alice's pure strategy $\tilde  A^a$. 
He aborts the protocol if this measurement detects the final state is not consistent with Alice's pure strategy $\tilde  A^a$. If Alice is honest, he never aborts. 
\end{itemize} 

\noindent
Protocols are designed with the intention to achieve the following two important properties of interest:
\begin{itemize}
\item Binding: Alice cannot change her mind after the Commit Phase and reveal the other value of $a$ (without being detected by Bob). 
\item Concealing: Bob cannot learn Alice's bit $a$ before she reveals it during the Reveal Phase. 
\end{itemize} 
 
The references \cite{May97, LC97, LC97a} showed that when Alice and Bob's actions are known to both parties, bit-commitment with perfect binding and concealing is impossible. 
In the more general setting when the actions need not be fully specified beforehand, bit-commitment was shown to be impossible in \cite{10.1063/1.1904510} for the channel setting, and in \cite{DArianoKSW07, ChiribellaD+09b} for the interactive setting.  
Here, we give another proof of this fact which follows straightforwardly from the properties of the strategy fidelity and strategy norm which we have already discussed. 
  
We define the cheating probabilities of Alice and Bob as follows:
\begin{center}
\begin{tabularx}{\textwidth}{rX}
  $\PBBC$: & The maximum probability with which a dishonest Bob can cheat by \emph{learning}  an honest Alice's committed bit $a \in \{ 0, 1 \}$ after the Commit Phase. \\
$\PABC$: & The maximum probability with which a dishonest Alice can cheat by \emph{changing} her commitment from $0$ to $1$ (or from $1$ to $0$) after the Commit Phase such that Bob accepts the new value (i.e., he does not abort). 
\end{tabularx} 
\end{center} 

\begin{remark} \label{rem13}
Note that in the definition of cheating Alice above, we do not assume Alice knows Bob's actions. It could even be the case that Bob's sole purpose is to choose a co-strategy such as to minimize $\PABC$. 
\end{remark} 
 
Cheating Bob wishes to distinguish between one of two uniformly randomly chosen strategies. We know from~\cite{Gutoski12} that 
\begin{equation}
\PBBC = \half + \frac{1}{4} \snorm{A^0-A^1}{\mathrm{r}}.
\end{equation}

In Section~\ref{Crypto}, we show that 
\begin{equation}
 \PABC \geq \fid_{r} (A^0, A^1)^2.
 \end{equation} 
An interesting observation is that this only depends on Alice's honest strategies, not Bob's. 

Thus, by the Fuchs-van de Graaf inequalities for strategies (Proposition~\ref{prop:basicProp}),  
 we have the following trade-off lower bound. 

\begin{theorem}  
In any interactive quantum protocol for bit-commitment, we have that 
\begin{equation}
 \sqrt{\PABC} + 2 \PBBC \geq 2 
 \end{equation} 
implying 
\begin{equation} 
\max \{ \PABC, \PBBC \} \geq \frac{9-\sqrt{17}}{8}\approx 61 \% 
\end{equation} 
(recall the definitions of $\PABC$ and $\PBBC$ above Remark~\ref{rem13}). 
In other words, at least one of Alice or Bob can successfully cheat with probability at least $61 \%$ making bit-commitment insecure. 
\end{theorem} 

Note that this is a similar bound to the one obtained in \cite{ChiribellaD+09b} for the interactive setting and exactly the same as in~\cite{10.1063/1.1904510} in the channel setting. 

We remark that, in the scenario when Alice and Bob's actions are completely specified, optimal protocols are known~\cite{ChaillouxK11} (albeit with a slightly different definition of cheating Alice). We leave it as an open problem to determine if the bound we present above is optimal in the scenario when Bob's actions are not specified. Moreover, it would be interesting to see whether the two scenarios share the same optimal cheating probabilities.  

\paragraph{1-out-of-2 interactive oblivious string transfer} \quad \\ 

This is an interactive cryptographic task between Alice and Bob where Bob has two bit-strings\footnote{The bit-length of the strings are, surprisingly, not important for the purposes of this paper.} $(x_0, x_1)$ and Alice wishes to learn one of the two in the following manner: 
\begin{itemize}
\item Alice chooses a uniformly random bit $a$ which corresponds to her choice of which string she wishes to learn, and interacts with Bob via the $r$-round pure strategy $\tilde A^a$. 
\item For every $(x_0, x_1)$, Bob uses a pure co-strategy $\tilde B^{x_0, x_1}$, such that Alice learns the string $x_a$ with certainty by measuring her private space $\cZ_r$ at the end of the protocol. 
\end{itemize} 

Note that we do not assume any structure on how Bob behaves other than the consistency condition above. For example, $x_0$ and $x_1$ may be the result of another protocol of which Alice is not part, and thus she does not even know the distribution from which they are drawn. 
Again, Bob's strategy may be such that, conditioned on the above requirements, he just wants to foil Alice's cheating, as defined below. 

We define the cheating probabilities of Alice and Bob as follows:
\begin{center}
\begin{tabularx}{\textwidth}{rX}
  $\PBOT$: & The maximum probability with which a dishonest Bob can cheat by correctly \emph{learning} an honest Alice's choice bit $a$. \\
$\PAOT$: & The maximum probability with which a dishonest Alice can cheat by correctly \emph{learning} $x_0$ after learning $x_1$ with certainty, or vice versa. 
\end{tabularx}
\end{center}

Cheating Bob behaves the exact same as in a bit-commitment protocol. Thus his cheating probability is again
\begin{equation}
 \PBOT = \frac{1}{2} + \frac{1}{4} \snorm{A^0-A^1}{\mathrm{r}}.
 \end{equation}

In Section~\ref{Crypto}, we show the following bound on cheating Alice:
\begin{equation}
 \PAOT \geq \fid_r(A^0, A^1)^2.
 \end{equation} 

This yields the same bound as in bit-commitment, below. 

\begin{theorem} 
In any interactive quantum protocol for 1-out-of-2 oblivious string transfer, we have that 
\begin{equation}
 \sqrt{\PAOT} + 2 \PBOT    \geq 2 
 \end{equation} 
implying 
\begin{equation} 
\max \{ \PAOT, \PBOT \} \geq \frac{9-\sqrt{17}}{8}\approx 61 \% 
\end{equation}
(recall the definitions of $\PAOT$ and $\PBOT$ above). 
In other words, at least one of Alice or Bob can successfully cheat with probability at least $61 \%$, making oblivious string transfer insecure.   
\end{theorem}  

Note that in the case where Bob has two \emph{bits} (i.e., the strings have bit-length $1$), an optimal security trade-off between Alice and Bob is known~\cite{CGS16}:
\begin{equation}
 \PAOT + 2 \PBOT \geq 2. 
 \end{equation}
However, this assumes perfect knowledge of Alice and Bob's honest strategies. Thus, our bound for cheating Alice is a bit weaker, but has the added benefit of only depending on her honest strategies. 
  
\subsection{Paper organization} 
We start with presenting some technical lemmas involving the strategy fidelity and generalizing Uhlmann's Theorem in Section~\ref{sec:lemma}. In Section~\ref{sect:mon}, we prove the monotonicity of the strategy fidelity under the action of supermaps. We then use the technical lemmas to formulate the strategy fidelity of two strategies as a semidefinite program in Section~\ref{sect:sdp}. We conclude the paper by presenting our main application of the study of strategy fidelity which is to capture Alice's cheating probability in interactive bit-commitment and oblivious string transfer, discussed in Section~\ref{Crypto}. 
  
\section{Technical lemmas and the strategy generalization of Uhlmann's Theorem} 

\label{sec:lemma}

In this section we prove two lemmas that allow us to establish nontrivial properties of the strategy fidelity.
These lemmas are used to prove the strategy generalization of Uhlmann's Theorem (Theorem~\ref{thm:StrUhlmann})  
and to provide a semidefinite programming formulation of the strategy fidelity (Theorem~\ref{SDP}). 

Before we proceed, let us introduce some notation. 
Let 
$\kprod{\cY}{i}{j}\kprod{\cX}{i'}{j'}$ be short for $\kprod{\cY}{i}{j}\ot\kprod{\cX}{i'}{j'}$.
Let $\lin{\cX}$, $\uni{\cX}$, $\her{\cX}$, $\pos{\cX}$, and $\den{\cX}$ be, respectively, the set of all linear, unitary, Hermitian, positive semidefinite, and density operators acting on $\cX$. 
Let $\uball{\cX}$ be the convex hull of $\uni{\cX}$, namely, the set of all operators 
$K\in\lin{\cX}$ such that $\norm{K}\leq 1$. 
Suppose $\cX$ and $\cY$ are two complex Euclidean spaces with fixed standard basis.
Given a linear operator $A:\cX\to\cY$ written in the standard basis as
\begin{equation}
A=\sum_{i=1}^{\dim(\cX)}
\sum_{j=1}^{\dim(\cY)}
 a_{j,i}\ket{j}\bra{i},
\end{equation} the \emph{vectorization} of $A$ is 
\begin{equation}
\vecket{A} :=
\sum_{i=1}^{\dim(\cX)}
\sum_{j=1}^{\dim(\cY)}
 a_{j,i}\ket{j}\ot\ket{i} \in\cY\otimes\cX 
\end{equation}
and its adjoint is $\vecbra{A}:=\Pa{\vecket{A}}^*$. 
  
\begin{lemma}[Inner product is linear in $B$]
\label{lm:linear-in-B}

Let $S,T$ be $r$-round strategies and let $\tilde S,\tilde T$ be any purifications of $S,T$.
Let $B$ be a compatible $r$-round co-strategy and let $\tilde B$ be any purification of $B$.
Let $\ket{\psi(\tilde S,\tilde B)},\ket{\psi(\tilde T,\tilde B)}$ be as in Definition~\ref{def:purestrat}
and let $K\in\lin{\cZ_r}$.
It holds that
\begin{equation}
\bra{\psi(\tilde S,\tilde B)} \Pa{K\ot I_{\cW_r}}\ket{\psi(\tilde T,\tilde B)} =
\vecbra{\tilde S} \Pa{K \ot B} \vecket{\tilde T}. 
\end{equation}

\end{lemma}

Note that the inner product above depends on $B$ but not on its purification $\tilde B$. This exemplifies what we stated earlier as Remark~\ref{rem:UniInvariance}.

\begin{proof}[Proof of Lemma~\ref{lm:linear-in-B}]
The proof mirrors that of Ref.\ \cite[Theorem 5]{GutoskiW07}.
The main difference is that here we compute an inner product between two \emph{distinct} vectors 
\begin{equation}
\ket{\psi(\tilde S,\tilde B)} 
\quad \text{and} \quad
(K\otimes I_{\cW_r})\ket{\psi(\tilde T,\tilde B)}
\end{equation} 
(both being \emph{normalized} if $K$ is unitary) arising from the distinct pure strategies $\tilde S,\tilde T$ for Alice, whereas the proof of Ref.\ \cite[Theorem 5]{GutoskiW07} computes a similar inner product between two \emph{identical, subnormalized} vectors.
Further clarification of that proof is given in Ref.\ \cite{Gutoski-Phd}; we draw upon both of the references \cite{GutoskiW07,Gutoski-Phd} for the present proof.

It was proved in \cite{GutoskiW07} that
\begin{align}
  \ket{\psi(\tilde S,\tilde B)} &= \Pa{ \vecbra{I_{\kprod{\cY}{1}{r}\kprod{\cX}{1}{r}}} \ot
   I_{\cZ_r\cW_r} } \Pa{\vecket{\tilde S} \ot \vecket{\tilde B}}, \\
  \ket{\psi(\tilde T,\tilde B)} &= \Pa{ \vecbra{I_{\kprod{\cY}{1}{r}\kprod{\cX}{1}{r}}} \ot
   I_{\cZ_r\cW_r} } \Pa{\vecket{\tilde T} \ot \vecket{\tilde B}},
\end{align}
from which we obtain
\begin{multline} \label{eq:lm-big}
  \bra{\psi(\tilde S,\tilde B)}\Pa{K\ot I_{\cW_r}}\ket{\psi(\tilde T,\tilde B)} \\ = \Pa{\vecbra{\tilde S} \ot \vecbra{\tilde B}}
    \Pa{ \vecket{I_{\kprod{\cY}{1}{r}\kprod{\cX}{1}{r}}}\vecbra{I_{\kprod{\cY}{1}{r}\kprod{\cX}{1}{r}}} \ot K\ot I_{\cW_r} }
    \Pa{\vecket{\tilde T} \ot \vecket{\tilde B}}. 
\end{multline}
Let
\begin{equation}
K=\sum_{i,i'=1}^{\dim(\cZ_r)} k_{i,i'}\ket{i}\bra{i'}
\end{equation}
and, for each $i=1,\dots,\dim(\cZ_r)$ and $j=1,\dots,\dim(\cW_r)$, let
\begin{align}
  \tilde S_i,\tilde T_i &: \kprod{\cX}{1}{r}\to\kprod{\cY}{1}{r} \\
  \tilde B_j &: \kprod{\cY}{1}{r}\to\kprod{\cX}{1}{r}
\end{align}
be the operators satisfying
\begin{equation}
  \tilde S = \sum_{i=1}^{\dim(\cZ_r)} \tilde S_i \ot \ket{i}, \qquad
  \tilde T = \sum_{i=1}^{\dim(\cZ_r)} \tilde T_i \ot \ket{i}, \qquad
  \tilde B = \sum_{j=1}^{\dim(\cW_r)} \tilde B_j \ot \ket{j}
\end{equation}
so that \eqref{eq:lm-big} becomes
\begin{equation} \label{eq:lm-big2}
  \sum_{i,i'=1}^{\dim(\cZ_r)} \sum_{j=1}^{\dim(\cW_r)} k_{i,i'} \Pa{ \vecbra{\tilde S_i} \ot \vecbra{\tilde B_j} } \vecket{I_{\kprod{\cY}{1}{r}\kprod{\cX}{1}{r}}} \vecbra{I_{\kprod{\cY}{1}{r}\kprod{\cX}{1}{r}}} \Pa{ \vecket{\tilde T_{i'}} \ot \vecket{\tilde B_j} }.
\end{equation}
Using an identity from Ref.\ \cite[Proposition 3.5]{Gutoski-Phd} we have that \eqref{eq:lm-big2} becomes 
\begin{equation}
  \sum_{i,i'=1}^{\dim(\cZ_r)} \sum_{j=1}^{\dim(\cW_r)} k_{i,i'} \vecbraket{\tilde S_i}{\tilde B_j^*} \cdot \vecbraket{\tilde B_j^*}{\tilde T_{i'}}
  = \sum_{i,i'=1}^{\dim(\cZ_r)} k_{i,i'} \vecbra{\tilde S_i} B \vecket{\tilde T_{i'}}
   = \vecbra{\tilde S}\Pa{K \ot B} \vecket{\tilde T}
\end{equation}
given that Bob's co-strategy equals $B=\sum_j \veckb{\tilde B_j^*}$ as observed in Ref.\ \cite[Theorem 3.1]{Gutoski-Phd}.  
\end{proof}
 
\begin{lemma} \label{lem:linear-in-B}
Let $S,T$ be $r$-round strategies and let $\tilde S,\tilde T$ be any purifications of $S,T$. It holds that
\begin{equation}
\rfid(S,T) = \max_K \min_B \; \Re \, \left( {\vecbra{\tilde S}\Pa{K\ot B}\vecket{\tilde T}} \right)  
\end{equation}
where the minimum is over all compatible $r$-round co-strategies $B$ for Bob and the maximum is over all $K \in\uball{\cZ_r}$ acting on the final memory space $\cZ_r$ for Alice. 
\end{lemma}

\begin{proof}
By applying Lemma \ref{lm:linear-in-B} to Eqn.\ \eqref{eq:naive-min-max}, we get
\begin{equation} \label{eq:minmax-1}
\rfid(S,T) 
= \min_B \max_U \Abs{ \vecbra{\tilde S} \Pa{U \ot B} \vecket{\tilde T} }
= \min_B \max_U \; \Re\Pa{ \vecbra{\tilde S} \Pa{U \ot B} \vecket{\tilde T} },
\end{equation}
where the maximum is over all $U \in\uni{\cZ_r}$.
The advantage of this identity is that the objective function is linear in $U$.
Since linear functions are also convex and since the maximum of a convex function over a compact convex set is always achieved at an extreme point, the above quantity does not change if we replace the maximization over unitaries with the maximization over the convex hull of the unitaries. Namely,
\begin{equation} \label{eq:minmax-3}
\rfid(S,T)
= \min_B \max_K \; \Re\Pa{ \vecbra{\tilde S} \Pa{K \ot B} \vecket{\tilde T} },
\end{equation}
where the maximization is over all $K\in\uball{\cZ_r}$.
By a standard min-max theorem from convex analysis (see, for example, \cite{Rockafellar70}), we may reverse the order of optimization, concluding the proof.  
\end{proof}
 
Now, with Lemmas \ref{lm:linear-in-B} and \ref{lem:linear-in-B} at our disposal, we proceed to prove the strategy generalization of Uhlmann's Theorem.
 
\begin{proof}[Proof of Theorem~\ref{thm:StrUhlmann}]
From Lemma~\ref{lem:linear-in-B}, it follows that
\begin{equation} 
\rfid(S,T) \leq \max_K \min_B \Abs{ \vecbra{\tilde S} \Pa{K \ot B} \vecket{\tilde T}}. 
\end{equation} 
We square this inequality and apply Lemma \ref{lm:linear-in-B} to obtain
\begin{equation} \label{eq:minmax-5}
\rfid(S,T)^2
\leq \max_K \min_B \, 
\bra{\psi(\tilde S,\tilde B)} \Pa{K\ot I_{\cW_r}}\kb{\psi(\tilde T,\tilde B)}\Pa{K^*\ot I_{\cW_r}} \ket{\psi(\tilde S,\tilde B)}.
\end{equation}
Let us define $\bar{K}=\sqrt{I_{\cZ_r}-K^*K}$ (noting that $K^*K\preceq I_{\cZ_r}$) and 
\begin{equation}
\Xi_K : \lin{\cZ_r}\to\lin{\cZ_r} : X \mapsto KXK^*+\bar{K}X\bar{K}^*,
\end{equation}
which is a quantum channel as its Kraus representation $\{K,\bar{K}\}$ satisfies $K^*K+\bar{K}^*\bar{K}=I_{\cZ_r}$.
Since
\begin{equation} \label{eq:minmax-7}
\bra{\psi(\tilde S,\tilde B)} \Pa{\bar{K}\ot I_{\cW_r}}\kb{\psi(\tilde T,\tilde B)}
 \Pa{\bar{K}^*\ot I_{\cW_r}} \ket{\psi(\tilde S,\tilde B)} \geq 0
\end{equation}
for all $K$ and all $\tilde B$, we have
\begin{eqnarray}  
\rfid(S,T)^2
& \leq \displaystyle\max_K \displaystyle\min_B \, 
\bra{\psi(\tilde S,\tilde B)} \Br{ \Pa{\Xi_K\otimes I_{\lin{\cW_r}}}\Pa{\kb{\psi(\tilde T,\tilde B)}} } \ket{\psi(\tilde S,\tilde B)} \\
& \leq \displaystyle\max_\Xi \displaystyle\min_B \, 
\bra{\psi(\tilde S,\tilde B)} \Br{ \Pa{\Xi\otimes I_{\lin{\cW_r}}}\Pa{\kb{\psi(\tilde T,\tilde B)}} } \ket{\psi(\tilde S,\tilde B)}. \label{eqn:IneqEq} 
\end{eqnarray}
However, we clearly have 
\begin{equation}
\rfid(S,T)^2
=
\min_B \max_\Xi \,  
\bra{\psi(\tilde S,\tilde B)} \Br{ \Pa{\Xi\otimes I_{\lin{\cW_r}}}\Pa{\kb{\psi(\tilde T,\tilde B)}} } \ket{\psi(\tilde S,\tilde B)} 
\end{equation}
due to Eqn.\ \eqref{eq:naive-min-max} and the fact that Uhlmann's Theorem also holds replacing unitaries with channels.
Hence, the inequality \eqref{eqn:IneqEq} is in fact an equality due to the max--min inequality. 
\end{proof}

\section{Monotonicity} 
\label{sect:mon}
 
Recall that strategy supermaps $\Upsilon$ map $r$-round strategies to $r'$-round strategies and they are linear and completely positive.
For our results, we need certain properties of the adjoints of strategy supermaps. 
To this end, we first prove the following lemma. 

\begin{lemma} \label{useful}
If $X \in \pos{\kprod{\cY}{1}{r}\kprod{\cX}{1}{r}}$ satisfies $\inner{X}{S} = 1$ for all
$r$-round strategies $S$ having input spaces $\cX_1,\ldots,\cX_r$ and output spaces $\cY_1,\ldots,\cY_r$, then $X$ is an $r$-round co-strategy having input spaces $\cY_1,\ldots,\cY_r$ and output spaces $\cX_1,\ldots,\cX_r$. 
\end{lemma} 

\begin{proof} 
If $X \succeq 0$ satisfies $\inner{X}{S} = 1$ for all strategies $S$, then $X$ also satisfies $\inner{X}{S'} \leq 1$ for all $S'$ such that $0 \preceq S' \preceq S$ for some strategy $S$. Thus, from Ref.~\cite{GutoskiW07}%
\footnote{In the terminology of \cite{GutoskiW07}, we have that $X \in (\downarrow\cS_r(\kprod{\cX}{1}{r},\kprod{\cY}{1}{r}))^\circ$.}, we have that there exists a co-strategy $B$ such that $X \preceq B$. 
For any pair of compatible strategy $S$ and co-strategy $B$, we have $\inner{B}{S}=1$~\cite{GutoskiW07}, therefore, we have 
\( \inner{X}{S} = \inner{B}{S}  \) 
for all strategies $S$. Next, if we consider
\begin{equation}
S = \frac{1}{\dim (\kprod{\cY}{1}{r})} I_{\kprod{\cY}{1}{r}\kprod{\cX}{1}{r}},
\end{equation}
where $\dim (\kprod{\cY}{1}{r})$ is the product of the dimensions of spaces $\cY_1,\ldots,\cY_r$, then this is a valid strategy. 
Then we get that $X$ and $B$ have the same trace. Since $0 \preceq X \preceq B$, we have that $X = B$, which completes the proof.
\end{proof}  
 
We can now provide an important property of the adjoint of strategy supermaps. 

\begin{lemma} \label{lem:UpsilonStar}
If $\Upsilon$ is a strategy supermap from $r$-round strategies to $r'$-round strategies, then $\Upsilon^*$ is a co-strategy supermap\footnote{Here, we define co-strategy supermaps in the analogous way as strategy supermaps.} from $r'$-round co-strategies to $r$-round co-strategies.  
\end{lemma} 

\begin{proof} 
Let $B$ be an $r'$-round co-strategy. Then have have that 
\begin{equation}
 \inner{\Upsilon^*(B)}{S} = \inner{B}{\Upsilon(S)} = 1
 \end{equation}
for all $r$-round strategies $S$. Since $\Upsilon$ is completely positive, so is $\Upsilon^*$, implying that $\Upsilon^*(B)$ is positive semidefinite. From Lemma~\ref{useful}, we have that $\Upsilon^*(B)$ is an $r$-round co-strategy, as required. 
\end{proof} 
 
\subsection{Monotonicity of the strategy fidelity}

We now provide a proof of Theorem~\ref{mon-fid}. 

\begin{proof}[Proof of Theorem~\ref{mon-fid}] 
Since $\Upsilon$ is completely positive, we can let
\begin{alignat}{3}
  &\Upsilon&&:S&&\mapsto \ptr{\cM}{MSM^*} \\
  &\Upsilon^*&&:S'&&\mapsto M^*(I_\cM\ot S')M
\end{alignat}
be Stinespring representations of $\Upsilon$ and its adjoint $\Upsilon^*$, respectively, where the operator $M$ has the form
\begin{equation}
  M:\kprod{\cY}{1}{r}\kprod{\cX}{1}{r}\to\kprod{\cY'}{1}{r'}\kprod{\cX'}{1}{r'}\cM
\end{equation}
for some appropriately large space $\cM$.
Let
\begin{equation}
  \tilde S,\tilde T : \kprod{\cX}{1}{r} \to \kprod{\cY}{1}{r}\cZ_r
\end{equation}
be purifications of the $r$-round strategies $S,T$ for some appropriately large final memory space $\cZ_r$,
and let $\vecket{\tilde S},\vecket{\tilde T}$ be their respective vectorizations.
Given the standard basis $\{\ket{i}: i \in \{ 1, \ldots, \dim (\kprod{\cX}{1}{r}) \} \}$ of 
$\kprod{\cX}{1}{r}$, we have
\begin{equation}
S = \sum_{i,j} \Ptr{\cZ_r}{\tilde S\ket{i}\bra{j}\tilde S^*}\ot \ket{i}\bra{j} 
 = \mathrm{Tr}_{\cZ_r} \bigg({\bigg( \sum_i \tilde S\ket{i}\ot\ket{i} \bigg) \bigg( \sum_j\bra{j}\tilde S^*\ot \bra{j} \bigg)\bigg)}
 = \Ptr{\cZ_r}{\veckb{\tilde S}}.
 \end{equation}
Hence
\begin{equation}
\Upsilon(S)=\Ptr{\cM}{MSM^*}
=\Ptr{\cZ_r\cM}{(M\ot I_{\cZ_r})\veckb{\tilde S}(M^*\ot I_{\cZ_r})},
\end{equation}
and an analogous equality holds for $T$ and $\tilde T$.  
Thus one can observe that the vectors
\begin{equation}
  (M\ot I_{\cZ_r})\vecket{\tilde S}, (M\ot I_{\cZ_r})\vecket{\tilde T} \in \kprod{\cY}{1}{r}\kprod{\cX}{1}{r}\cZ_r\cM
\end{equation}
are the vectorizations of purifications of the $r'$-round strategies $\Upsilon(S),\Upsilon(T)$ with final memory space $\cZ_r\cM$. 

By Eqn.~\eqref{eq:naive-min-max} and Lemma \ref{lm:linear-in-B} we have 
\begin{equation} \label{eq:mono-1}
  \fid_{\mathrm{r'}}(\Upsilon(S),\Upsilon(T)) = \min_{B'} \max_{U'} \Abs{ \vecbra{\tilde S}(M^*\ot I_{\cZ_r}) \Pa{U' \ot B'} (M\ot I_{\cZ_r})\vecket{\tilde T} }
\end{equation}
where the minimum is over all $r'$-round co-strategies $B'$ for Bob and the maximum is over all unitaries $U'\in\uni{\cZ_r\cM}$ on the final memory space $\cZ_r\cM$ for Alice.
The quantity \eqref{eq:mono-1} can only decrease if we restrict the domain of maximization to unitaries of the form $U\ot I_\cM$ for some $U\in\uni{\cZ_r}$, thus 
\begin{align}
\fid_{\mathrm{r'}}(\Upsilon(S),\Upsilon(T)) &\geq 
  {\min_{B'} \max_U} \Abs{ \vecbra{\tilde S}(M^*\ot I_{\cZ_r}) \Pa{U \ot I_\cM \ot B'} (M\ot I_{\cZ_r})\vecket{\tilde T} } \\
  &= {\min_{B'} \max_U} \Abs{ \vecbra{\tilde S}\Pa{U \ot M^*(I_\cM \ot B')M} \vecket{\tilde T} } \\
  &= {\min_{B'} \max_U} \Abs{ \vecbra{\tilde S}\Pa{U \ot \Upsilon^*(B')} \vecket{\tilde T} }. \label{eq:mono-2}
\end{align}
As the image under $\Upsilon^*$ of the set of all $r'$-round co-strategies is a subset of the set of all $r$-round co-strategies (by Lemma~\ref{lem:UpsilonStar}), the quantity \eqref{eq:mono-2} can only decrease if we extend the domain of minimization to all $r$-round co-strategies $B$ for Bob:
\begin{equation}
\fid_{\mathrm{r'}}(\Upsilon(S),\Upsilon(T)) \geq \min_B \max_U \Abs{ \vecbra{\tilde S}\Pa{U \ot B} \vecket{\tilde T} } = \rfid(S,T)
\end{equation}
as desired. 
\end{proof}

\subsection{Monotonicity of the strategy norm}
  
\begin{proof}[Proof of Theorem~\ref{mon-norm}] 
By the definition of the strategy norm, Definition~\ref{def:str_norm}, we have 
\begin{align}
  \Snorm{\Upsilon(H)}{\mathrm{r'}}
  &= \max\Set{\Inner{B'_0-B'_1}{\Upsilon(H)} : \textnormal{$B'_0 + B'_1$ is an $r'$-round co-strategy, $B'_0, B'_1 \succeq 0$}} \\
  &= \max\Set{\Inner{\Upsilon^*(B'_0)-\Upsilon^*(B'_1)}{H} : \textnormal{$B'_0 + B'_1$ is an $r'$-round co-strategy, $B'_0, B'_1 \succeq 0$}} \\ 
  &\leq \max\Set{\Inner{B_0-B_1}{H} : \textnormal{$B_0 + B_1$ is an $r$-round co-strategy, $B_0, B_1 \succeq 0$}} \\
  &= \Snorm{H}{\mathrm{r}}.   
\end{align}  
Note that $\Upsilon^*$ is both linear and completely positive.
Thus, given $B'_0,B'_1\succeq0$ such that $B'_0+B'_1$ is an $r'$-round co-strategy, we have that $B_0:=\Upsilon^*(B'_0)\succeq0$ and $B_1:=\Upsilon^*(B'_1)\succeq0$ and,  by Lemma~\ref{lem:UpsilonStar}, $B_0+B_1$ is an $r$-round co-strategy.
But the image under $\Upsilon^*$ of the set of all $r'$-round co-strategies may be a strict subset of the set of all $r$-round co-strategies, hence the inequality in the above expression.  
\end{proof}   
 
\section{Semidefinite programming formulation for strategy fidelity} 
\label{sect:sdp} 
   
In this section, we use Lemma~\ref{lem:linear-in-B} 
to prove Theorem~\ref{SDP}. From Lemma~\ref{lem:linear-in-B}, we have that  
\begin{equation}
 \rfid(S,T)^2 = {\max} \; \{ \phi(K) : K\in\uball{\cZ_r} \}
 \end{equation} 
where $\phi(K) := {\displaystyle\min_B} \; \Re \, \vecbra{\tilde S} \Pa{K \ot B} \vecket{\tilde T}$, and $B$ is Bob's co-strategy.  
By defining 
\begin{equation}
 C := \frac{1}{2} \Ptr{\cZ_r}{\Pa{K \otimes I_{\kprod{\cY}{1}{r}\kprod{\cX}{1}{r}}} \vecket{\tilde T} \vecbra{\tilde S}} + \frac{1}{2} \Br{\Ptr{\cZ_r}{\Pa{K \otimes  I_{\kprod{\cY}{1}{r}\kprod{\cX}{1}{r}}} \vecket{\tilde T} \vecbra{\tilde S}}}^*
 \end{equation} 
we can write 
\begin{equation}
 \phi(K) = {\min_B} \, \inner{C}{B}.
 \end{equation} 
From \cite[Corollary 7]{GutoskiW07}, we know that $B$ must satisfy $B = Q_{r} \otimes I_{\cY_{r}}$ for some $(Q_1, \ldots, Q_r)$ satisfying 
\begin{equation}
\tr{Q_{1}} = 1, 
\quad 
\ptr{\cX_i}{Q_{i}} = Q_{i-1} \otimes I_{\cY_{i-1}}, \textup{ for } i \in \{ 2, \ldots, r \} 
\end{equation}
and $Q_1 \in \pos{\cX_1}$, $Q_i \in \pos{\kprod{\cY}{1}{i-1}\ot\kprod{\cX}{1}{i}}$, for $i \in \{ 2, \ldots, r \}$. 
Thus, $\phi(K)$ can be formulated as a semidefinite program. Its dual can be written as 
\begin{multline}
 \alpha(K) := \max 
\Big\{ \; t \; : \; t I_{\cX_1} \preceq \ptr{\cY_1}{R_1}, \\ 
R_j \otimes I_{\cX_{j+1}} \preceq \ptr{\cY_{j+1}}{R_{j+1}} \textup{ for } j \in \{ 1, \ldots, r -1 \}, \; 
R_{r} \preceq C \; \Big\},
\end{multline} 
where $R_j \in \her{\kprod{\cY}{1}{j}\ot\kprod{\cX}{1}{j}}$. 
Since this has a strictly feasible solution, as does the primal, we know $\alpha(K) = \phi(K)$ by strong duality and $\alpha(K)$ attains an optimal solution.
We now let $M = \left[ \begin{array}{cc} 
I_{\cZ_r} & K \\ 
K^* & I_{\cZ_r} 
\end{array} \right]$ and set $M \succeq 0$ to get $\| K \| \leq 1$. We can check that $C$ is a linear function in $M$ (since $M$ is Hermitian). Thus, we have that the strategy fidelity can be written as in Theorem~\ref{SDP}. 
 
\section{Alice's cheating in interactive bit-commitment and oblivious string transfer} 
\label{Crypto}

In this section we show that Alice can cheat with probability $\fid_r(A^0, A^1)^2$ in either bit-commitment or oblivious string transfer. 
The cheating has the same flavour in both cases: Alice will follow the protocol honestly, then try to change her state as to make it look like she chose the other strategy from the beginning. Suppose Alice uses pure strategy $\tilde A^a$ and Bob uses pure co-strategy $\tilde B$. For brevity, define for each $a \in \{ 0, 1 \}$ the following states 
\begin{equation} \label{cheating}
\ket{\psi_a} := \ket{\psi(\tilde A^a,\tilde B)}
\quad \textup{ and } \quad 
\sigma_a := (\Xi^a \ot I_{\cW_r})(\kb{\psi_a})   
\end{equation} 
where $\Xi_a$ is the strategy fidelity-achieving channel (from Definition~\ref{def:fidchannel}) such that 
\begin{equation} 
\label{fidelitybound}
\bra{\psi_{\bar a}} 
\sigma_a 
\ket{\psi_{\bar a}} 
\geq 
\rfid(A^0, A^1)^2.  
\end{equation} 
Note that the aim of $\Xi^a$ is to get $\sigma_a$ as  close as possible to $\kb{\psi_{\bar{a}}}$. 

\subsection{Bit-commitment} 

When we study interactive bit-commitment, we are applying the strategy/co-strategy formalism \emph{only} to the Commit Phase. From the above discussion, Alice can create the state 
\begin{equation}
 \sigma_a \in \den{\cZ_r \otimes \cW_r}
 \end{equation} 
to try to change her commitment from $a$ to $\bar{a}$. 
Then Alice continues her actions to ``reveal'' $\bar a$ in the Reveal Phase, as does Bob (even though Bob's actions are not specified to Alice). We just assume that this entire process is done by a unitary $U_{\bar a}$ acting on $\cZ_r \otimes \cW_r$. Then, Bob has a projective measurement $\{ \Pi_{accept}, \Pi_{reject} \}$ which accepts $U_{\bar a} \ket{\psi_{\bar a}}$ with certainty, thus leading to a \emph{non-destructive measurement}. Thus, we have 
\begin{equation}
 (I_{\cZ_r} \otimes \Pi_{accept}) U_{\bar a} \ket{\psi_{\bar a}} = U_{\bar a} \ket{\psi_{\bar a}}. \end{equation} 
This implies that 
\begin{equation}
 (I_{\cZ_r} \otimes \Pi_{accept}) \succeq U_{\bar a} \kb{\psi_{\bar a}} U_{\bar a}^*.
 \end{equation} 
However, Alice's actions have led to them sharing $U_{\bar a} \sigma_a U_{\bar a}^*$ at the end of the protocol. So, we have that Alice successfully reveals $\bar a$ with probability 
\begin{equation}
\PABC \geq 
\inner{I_{\cZ_r} \otimes \Pi_{accept}}{U_{\bar a} \sigma_a U_{\bar a}^*} 
\geq 
\Inner{U_{\bar a} \kb{\psi_{\bar a}} U_{\bar a}^*}{U_{\bar a} \sigma_a U_{\bar a}^*} 
= 
\Inner{\kb{\psi_{\bar a}}}{\sigma_a} 
\geq 
\fid_r(A^0, A^1)^2  
 \end{equation} 
using Eqn.~\eqref{fidelitybound}, as desired. 
  
\subsection{Oblivious string transfer} 

We can assume Alice uses a projective measurement $\{ \Pi^a_{z} \}$ to learn her desired string. Note that since $x_a$ is learned with certainty, this is a \emph{non-destructive measurement}, as in the bit-commitment analysis above. That is, we have 
\begin{equation}
\Pa{\Pi^a_{x_a} \ot I_{\cW_r}} \ket{\psi(\tilde A^a, \tilde B^{x_0, x_1})} = \ket{\psi(\tilde A^a, \tilde B^{x_0, x_1})}  
\end{equation}
for all $a$ and $(x_0, x_1)$. Again, this implies 
\begin{equation} \label{neato}
\Pi^a_{x_a} \ot I_{\cW_r} \succeq \kb{\psi(\tilde A^a, \tilde B^{x_0, x_1})}. 
\end{equation} 

Thus, after learning $x_a$, she can create the state $\sigma_a$ (defined above) to try to learn $x_{\bar a}$. 
(Here, the $\tilde B$ in the definition of $\sigma_a$ is $\tilde B^{x_0, x_1}$.) Then she measures as if she had used pure strategy $\tilde A^{\bar{a}}$ (that is, using $\{ \Pi^{\bar{a}}_z \}$) to try to learn $x_{\bar a}$. Then, using \eqref{neato} and the definitions in \eqref{cheating}, we have
\begin{equation} 
\PAOT \geq 
\inner{\Pi^{\bar a}_{x_{\bar a}} \ot I_{\cW_r}}{\sigma_a}  
\geq 
\bra{\psi_{\bar a}} {\sigma_a} \ket{\psi_{\bar a}}  
\geq 
\rfid(A^0,A^1)^2,  
\end{equation} 
as desired. 
  
\section*{Acknowledgements} 

Research at the Perimeter Institute is supported by the Government of Canada through Industry Canada and
by the Province of Ontario through the Ministry of Research and Innovation. 
GG also acknowledges support from CryptoWorks21.
JS acknowledges support from NSERC Canada. 
Research at the Centre for Quantum Technologies at the National University of Singapore is
partially funded by the Singapore Ministry of Education and the National Research Foundation,
also through the Tier 3 Grant ``Random numbers from quantum processes,'' (MOE2012-T3-1-009).
This material is based on research supported in part by the Singapore National Research Foundation under NRF RF Award
No. NRF-NRFF2013-13.

\newcommand{\etalchar}[1]{$^{#1}$}


\begin{thebibliography}{ABDR04}

\bibitem[ABDR04]{ABD+04}
Andris Ambainis, Harry Buhrman, Yevgeniy Dodis, and Hein R{\"o}hrig.
\newblock Multiparty quantum coin flipping.
\newblock In {\em Proceedings of the 19th IEEE Annual Conference on
  Computational Complexity}, pages 250--259. IEEE Computer Society, 2004.
\newblock \href{https://doi.org/10.1109/CCC.2004.1313848}{DOI: 10.1109/CCC.2004.1313848}.

\bibitem[Amb01]{Amb01}
Andris Ambainis.
\newblock A new protocol and lower bounds for quantum coin flipping.
\newblock In {\em Proceedings of 33rd Annual ACM Symposium on the Theory of
  Computing}, pages 134 -- 142. ACM, 2001.
 \newblock \href{https://doi.org/10.1145/380752.380788}{DOI: 10.1145/380752.380788}.

\bibitem[BB84]{BB84}
Charles Bennett and Gilles Brassard.
\newblock Quantum cryptography: Public key distribution and coin tossing.
\newblock In {\em Proceedings of the IEEE International Conference on
  Computers, Systems, and Signal Processing}, pages 175--179. IEEE Computer
  Society, 1984.

\bibitem[BCF{\etalchar{+}}96]{BarnumC+96}
Howard Barnum, Carlton~M. Caves, Christopher~A. Fuchs, Richard Jozsa, and
  Benjamin Schumacher.
\newblock Noncommuting mixed states cannot be broadcast.
\newblock {\em Physical Review Letters}, 76:2818--2821, 1996.
\newblock \href{https://doi.org/10.1103/PhysRevLett.76.2818}{DOI:  10.1103/PhysRevLett.76.2818}.
\newblock arXiv:quant-ph/9511010.

\bibitem[BDR05]{10.1063/1.1904510}
Viacheslav~P. Belavkin, Giacomo~Mauro D'Ariano, and Maxim Raginsky.
\newblock Operational distance and fidelity for quantum channels.
\newblock {\em Journal of Mathematical Physics}, 46(6):062106, 2005.
\newblock \href{https://doi.org/10.1063/1.1904510}{DOI: 10.1063/1.1904510}.
\newblock arXiv:quant-ph/0408159.

\bibitem[CDP08a]{ChiribellaD+08f}
Giulio Chiribella, Giacomo~Mauro D'Ariano, and Paolo Perinotti.
\newblock Transforming quantum operations: Quantum supermaps.
\newblock {\em Europhysics Letters}, 83(3):30004, 2008.
\newblock arXiv:0804.0180 [quant-ph].
  
\bibitem[CDP08b]{ChiribellaD+08b}
Giulio Chiribella, Giacomo~Mauro D'Ariano, and Paolo Perinotti.
\newblock Memory effects in quantum channel discrimination. 
\newblock {\em Physical Review Letters}, 101:180501, 2008.
\newblock \href{https://doi.org/10.1103/PhysRevLett.101.180501}{DOI: 10.1103/PhysRevLett.101.180501}. 
\newblock arXiv:0803.3237 [quant-ph].
  
\bibitem[CDP09]{ChiribellaD+09a}
Giulio Chiribella, Giacomo~Mauro D'Ariano, and Paolo Perinotti.
\newblock Theoretical framework for quantum networks.
\newblock {\em Physical Review A}, 80(2):022339, 2009.
\newblock \href{https://doi.org/10.1103/PhysRevA.80.022339}{DOI: 10.1103/PhysRevA.80.022339}.
\newblock arXiv:0904.4483 [quant-ph].

\bibitem[CDP{\etalchar{+}}13]{ChiribellaD+09b}
Giulio Chiribella, Giacomo~Mauro D'Ariano, Paolo Perinotti, Dirk Schlingemann,
  and Reinhard~F. Werner.
\newblock A short impossibility proof of quantum bit commitment.
\newblock {\em Physics Letters A}, 377(15):1076--1087, 2013.
\newblock \href{https://doi.org/10.1016/j.physleta.2013.02.045}{DOI: 10.1016/j.physleta.2013.02.045}.
\newblock arXiv:0905.3801v1 [quant-ph].

\bibitem[CGS16]{CGS16}
Andr\'e Chailloux, Gus Gutoski, and Jamie Sikora.
\newblock Optimal bounds for semi-honest quantum oblivious transfer.
\newblock {\em Chicago Journal of Theoretical Computer Science}, (13), 2016.
\newblock \href{https://doi.org/10.4086/cjtcs.2016.013}{DOI: 10.4086/cjtcs.2016.013}.

\bibitem[CK09]{ChaillouxK09}
Andr\'e Chailloux and Iordanis Kerenidis.
\newblock Optimal quantum strong coin flipping.
\newblock In {\em Proceedings of the 50th IEEE Symposium on Foundations of
  Computer Science}, FOCS 2009, pages 527--533, 2009.
\newblock \href{https://doi.org/10.1109/FOCS.2009.71}{DOI: 10.1109/FOCS.2009.71}.
\newblock arXiv:0904.1511 [quant-ph].

\bibitem[CK11]{ChaillouxK11}
Andr\'e Chailloux and Iordanis Kerenidis.
\newblock Optimal bounds for quantum bit commitment.
\newblock In {\em Proceedings of the 52nd Annual IEEE Symposium on Foundations
  of Computer Science}, FOCS 2011, pages 354--362, 2011.
\newblock \href{https://doi.org/10.1109/FOCS.2011.42}{DOI: 10.1109/FOCS.2011.42}.
\newblock arXiv:1102.1678 [quant-ph].

\bibitem[CKS13]{ChaillouxKS13-QIC}
Andr\'e Chailloux, Iordanis Kerenidis, and Jamie Sikora.
\newblock Lower bounds for quantum oblivious transfer.
\newblock {\em Quantum Information and Computation}, 13(1\&2):158--177, 2013.
\newblock arXiv:1007.1875 [quant-ph].

\bibitem[CKS14]{ChaillouxKS13}
Andr\'e Chailloux, Iordanis Kerenidis, and Jamie Sikora.
\newblock Strong connections between quantum encodings, nonlocality, and
  quantum cryptography.
\newblock {\em Phys. Rev. A}, 89:022334, 2014.
\newblock \href{https://doi.org/10.1103/PhysRevA.89.022334}{DOI: 10.1103/PhysRevA.89.022334}.
\newblock arXiv:1304.0983 [quant-ph].
  
\bibitem[DKSW07]{DArianoKSW07} 
Giacomo Mauro D'Ariano, Dennis Kretschmann, Dirk Schlingemann, and Reinhard F. Werner. 
\newblock Reexamination of quantum bit commitment: The possible and the impossible. 
\newblock {\em Phys. Rev. A}, 76:032328, 2007. 
\newblock \href{https://doi.org/10.1103/PhysRevA.76.032328}{DOI: 10.1103/PhysRevA.76.032328}
\newblock arXiv:0605224 [quant-ph]. 
    
\bibitem[FvdG99]{FvdG99}  
Christopher A. Fuchs and Jeroen van de Graaf. 
\newblock Cryptographic distinguishability measures for quantum mechanical states.
\newblock {\em IEEE Transactions on Information Theory} 45(4):1216--1227, 1999. 
\newblock \href{https://doi.org/10.1109/18.761271}{DOI: 10.1109/18.761271}. 
    
\bibitem[Gut09]{Gutoski-Phd}
Gus Gutoski.
\newblock {\em Quantum strategies and local operations}.
\newblock PhD thesis, University of Waterloo, 2009.
\newblock arXiv:1003.0038 [quant-ph].

\bibitem[Gut12]{Gutoski12}
Gus Gutoski.
\newblock On a measure of distance for quantum strategies.
\newblock {\em Journal of Mathematical Physics}, 53(3):032202, 2012.
\newblock \href{https://doi.org/10.1063/1.3693621}{DOI: 10.1063/1.3693621}.
\newblock arXiv:1008.4636 [quant-ph].

\bibitem[GW07]{GutoskiW07}
Gus Gutoski and John Watrous.
\newblock Toward a general theory of quantum games.
\newblock In {\em Proceedings of the 39th ACM Symposium on Theory of Computing
  (STOC 2007)}, pages 565--574, 2007.
  \newblock \href{https://doi.org/10.1145/1250790.1250873}{DOI: 10.1145/1250790.1250873}.
\newblock arXiv:quant-ph/0611234.

\bibitem[Kit02]{Kitaev02}
Alexei Kitaev.
\newblock Quantum coin-flipping.
\newblock Presentation at the 6th Workshop on \emph{Quantum Information
  Processing} (QIP 2003), 2002.

\bibitem[KN04]{KerenidisN04}
Iordanis Kerenidis and Ashwin Nayak.
\newblock Weak coin flipping with small bias.
\newblock {\em Information Processing Letters}, 89(3):131--135, 2004.
\newblock \href{https://doi.org/10.1016/j.ipl.2003.07.007}{DOI: 10.1016/j.ipl.2003.07.007}.
\newblock arXiv:quant-ph/0206121.

\bibitem[LC97]{LC97}
Hoi-Kwong Lo and Hoi~Fung Chau.
\newblock Is quantum bit commitment really possible?
\newblock {\em Physical Review Letters}, 78(17):3410--3413, 1997.
\newblock \href{https://doi.org/10.1103/PhysRevLett.78.3410}{DOI: 10.1103/PhysRevLett.78.3410}.

\bibitem[LC98]{LC97a}
Hoi-Kwong Lo and Hoi~Fung Chau.
\newblock Why quantum bit commitment and ideal quantum coin tossing are
  impossible.
\newblock {\em Physica D: Nonlinear Phenomena}, 120(1--2):177--187, September
  1998.
\newblock Proceedings of the Fourth Workshop on Physics and Consumption.
\newblock \href{https://doi.org/10.1016/S0167-2789(98)00053-0}{DOI: 10.1016/S0167-2789(98)00053-0}.

\bibitem[May97]{May97}
Dominic Mayers.
\newblock Unconditionally secure quantum bit commitment is impossible.
\newblock {\em Physical Review Letters}, 78(17):3414--3417, 1997.
\newblock \href{https://doi.org/10.1103/PhysRevLett.78.3414}{DOI: 10.1103/PhysRevLett.78.3414}.

\bibitem[NS03]{NayakS03}
Ashwin Nayak and Peter Shor.
\newblock Bit-commitment based quantum coin flipping.
\newblock {\em Physical Review A}, 67(1):012304, 2003.
\newblock \href{https://doi.org/10.1103/PhysRevA.67.012304}{DOI: 10.1103/PhysRevA.67.012304}.
\newblock arXiv:quant-ph/0206123.

\bibitem[NST15]{NST15}
Ashwin Nayak, Jamie Sikora, and Levent Tun{\c{c}}el.
\newblock Quantum and classical coin-flipping protocols based on bit-commitment
  and their point games.
\newblock Available as arXiv.org e-Print quant-ph/1504.04217, 2015.

\bibitem[NST16]{NST16}
Ashwin Nayak, Jamie Sikora, and Levent Tun{\c{c}}el.
\newblock A search for quantum coin-flipping protocols using optimization
  techniques.
\newblock {\em Mathematical Programming}, 156(1):581--613, 2016.
\newblock \href{https://doi.org/10.1007/s10107-015-0909-y}{DOI: 10.1007/s10107-015-0909-y}.
 
\bibitem[NC00]{NC00}
Michael A. Nielsen and Isaac Chuang. 
\newblock Quantum Computation and Quantum Information. 
\newblock Cambridge
University Press, Cambridge, 2000.
 
\bibitem[Roc70]{Rockafellar70}
R.~Tyrrell Rockafellar.
\newblock {\em Convex Analysis}.
\newblock Princeton University Press, 1970.

\bibitem[{Rus}94]{Ruskai94}
M.~B. {Ruskai}.
\newblock Beyond strong subadditivity? {I}mproved bounds on the contraction of
  generalized relative entropy.
\newblock {\em Reviews in Mathematical Physics}, 6:1147--1161, 1994.
\newblock \href{https://doi.org/10.1142/S0129055X94000407}{DOI: 10.1142/S0129055X94000407}.

\bibitem[Sik17]{Sikora17}
Jamie Sikora.
\newblock Simple, near-optimal quantum protocols for die-rolling.
\newblock {\em Cryptography}, 1(2), 11, 2017.
\newblock \href{https://doi.org/10.3390/cryptography1020011}{DOI: 10.3390/cryptography1020011}. 

\bibitem[SR01]{SR01}
Robert~W. Spekkens and Terence Rudolph.
\newblock Degrees of concealment and bindingness in quantum bit commitment
  protocols.
\newblock {\em Physical Review A}, 65:012310, 2001.
\newblock \href{https://doi.org/10.1103/PhysRevA.65.012310}{DOI: 10.1103/PhysRevA.65.012310}.

\bibitem[Uhl76]{uhl76}
A.~Uhlmann.
\newblock The ``transition probability'' in the state space of a *-algebra.
\newblock {\em Reports on Mathematical Physics}, 9(2):273--279, 1976.
\newblock \href{https://doi.org/10.1016/0034-4877(76)90060-4}{DOI: 10.1016/0034-4877(76)90060-4}.

\bibitem[Wat09]{Watrous09}
John Watrous.
\newblock Semidefinite programs for completely bounded norms.
\newblock {\em Theory of Computing}, 5:217--238, 2009.
\newblock \href{http://dx.doi.org/10.4086/toc.2009.v005a011}{DOI: 10.4086/toc.2009.v005a011}.
\newblock arXiv:0901.4709v2 [quant-ph].

\bibitem[Wat13]{Watrous13}
John Watrous.
\newblock Simpler semidefinite programs for completely bounded norms.
\newblock {\em Chicago Journal of Theoretical Computer Science}, (8), 2013.
\newblock \href{http://dx.doi.org/10.4086/cjtcs.2013.008}{DOI: 10.4086/cjtcs.2013.008}.

\bibitem[Wie83]{Wiesner83}
Stephen Wiesner.
\newblock Conjugate coding.
\newblock {\em SIGACT News}, 15(1):78--88, January 1983.
\newblock \href{https://doi.org/10.1145/1008908.1008920}{DOI: 10.1145/1008908.1008920}.

\end{thebibliography}
\end{document}